\documentclass[11pt]{article}

\usepackage{amsfonts}
\usepackage{amssymb}
\usepackage{amstext}
\usepackage{amsmath}
\usepackage{enumerate}
\usepackage{algorithm}
\usepackage{algorithmic}

\usepackage{graphicx}

\usepackage{bbm}

\usepackage{amsthm}

\usepackage{thmtools}

\usepackage{verbatim} 

\usepackage[usenames,dvipsnames,svgnames,table]{xcolor}
\definecolor{dartmouthgreen}{rgb}{0.05, 0.5, 0.06}
\definecolor{ceruleanblue}{rgb}{0.16, 0.32, 0.75}

\usepackage[colorlinks=true,pdfpagemode=UseNone,citecolor=dartmouthgreen,linkcolor=ceruleanblue,urlcolor=BrickRed,pdfstartview=FitH]{hyperref}

\usepackage{framed}
\usepackage{enumitem}


\newtheorem{theorem}{Theorem}[section]
\newtheorem{conjecture}[theorem]{Conjecture}

\newtheorem{lemma}[theorem]{Lemma}
\newtheorem{definition}[theorem]{Definition}

\newtheorem{claim}[theorem]{Claim}

\newtheorem*{problem*}{Problem}
\newtheorem{remark}[theorem]{Remark}
\newtheorem*{remark*}{Remark}

\numberwithin{equation}{section}
\numberwithin{table}{section}

\newcommand{\R}{\ensuremath{\mathbb R}}

\newcommand{\E}[1]{{\mathbb{E}}\left[#1\right]}

\newcommand{\junk}[1]{}

\renewcommand{\L}{{\mathcal L}}
\newcommand{\vol}{{\rm vol}}

\newcommand{\vertiii}[1]{{\left\vert\kern-0.25ex\left\vert\kern-0.25ex\left\vert #1 \right\vert\kern-0.25ex\right\vert\kern-0.25ex\right\vert}}

\def\b1{{\bf 1}}
\def\eps{{\epsilon}}

\def\R{\mathbb{R}}

\def\vol{\operatorname{vol}} 
\def\diag{\operatorname{diag}} 

\global\long\def\E{\mathbb{E}}

\global\long\def\R{\mathbb{R}}

\newcommand{\inner}[2]{\langle #1, #2 \rangle} 

\DeclareMathOperator{\supp}{supp}

\title{On the Houdr\'e-Tetali conjecture about an isoperimetric constant of graphs} 

\author{
Lap Chi Lau\footnote{Cheriton School of Computer Science, University of Waterloo. Supported by NSERC Discovery Grant. 
},~~~~~
Dante Tjowasi\footnote{Cheriton School of Computer Science, University of Waterloo. Supported by NSERC Discovery Grant. 
}}

\date{}

\begin{document}

\begin{titlepage}
\def\thepage{}
\thispagestyle{empty}

\maketitle
	
	\begin{abstract}
		Houdr\'e and Tetali defined a class of isoperimetric constants $\varphi_p$ of graphs for $0 \leq p \leq 1$, and conjectured a Cheeger-type inequality for $\varphi_\frac12$ of the form 
		\[
		\lambda_2 \lesssim \varphi_\frac12 \lesssim \sqrt{\lambda_2},
		\]
		where $\lambda_2$ is the second smallest eigenvalue of the normalized Laplacian matrix.
		If true, the conjeecture would be a strengthening of the hard direction of the classical Cheeger's inequality.
		Morris and Peres proved Houdr\'e and Tetali's conjecture up to an additional log factor, using techniques from evolving sets.
		We present the following related results on this conjecture.
		\begin{enumerate}
			\item 
			We provide a family of counterexamples to the conjecture of Houdr\'e and Tetali, showing that the logarithmic factor is needed.
			\item
			We match Morris and Peres's bound using standard spectral arguments.
			\item
			We prove that Houdr\'e and Tetali's conjecture is true for any constant $p$ strictly bigger than $\frac12$,
			which is also a strengthening of the hard direction of Cheeger's inequality.
		\end{enumerate}
		Furthermore, our results can be extended to directed graphs using Chung's definition of eigenvalues for directed graphs.
	\end{abstract}

 \end{titlepage}
	
	\section{Introduction}
	Motivated by Talagrand's isoperimetric inequality for the hypercubes~\cite{Tal93} (see \autoref{sec:related}),
	Houdr\'e and Tetali~\cite{HT04} 
	extended Talagrand's isoperimetric constant to general Markov chains and also to different exponents. 
	
	\begin{definition}[Isoperimetric Constants for Markov Chains~\cite{HT04}]
		\label{def:iso}
		Let $(V,P,\pi)$ be an irreducible Markov chain with vertex set $V$, transition matrix $P \in \R^{|V| \times |V|}$ and stationary distribution $\pi : V \to \R_{\geq 0}$.
		For any $p \in (0,1]$, define the isoperimetric constant as
		\[
		\varphi_p(P) 
		:= \min_{S \subset V : \pi(S) \leq \frac12} \varphi_p(S) 
		:= \min_{S \subset V : \pi(S) \leq \frac12} \frac{\sum_{v \in S} \pi(v) \cdot \big( \sum_{u \in \overline{S}} P(v,u) \big)^p }{\pi(S)}.
		\]
		Let $\partial S := \{ v \in S \mid \sum_{u \in \overline{S}} P(v,u) > 0\}$ be the inner vertex boundary of $S$. 
		Then, for $p=0$, 
		\[
		\varphi_0(P) 
		:= \min_{S \subset V : \pi(S) \leq \frac12} \varphi_0(S)
		:= \min_{S \subset V : \pi(S) \leq \frac12} \frac{ \pi(\partial S) }{\pi(S)}.
		\]
		Given an undirected graph or a directed graph $G=(V,E)$, let $P_G$ be the transition matrix of the natural random walk, where $P_G(v,u) = 1/\deg(v)$ in the undirected case and $P_G(v,u) = 1/\deg^{\mathrm{out}}(v)$ in the directed case.
		Then the isoperimetric constant $\varphi_p(G)$ for the graph $G$ is defined as $\varphi_p(P_G)$.
	\end{definition}
	
	When $p=1$, this is known as the Cheeger isoperimetric constant of a Markov chain (see e.g.~\cite{LP17}) or the edge conductance/expansion of a graph (see \autoref{sec:prelim} for definition).
	When $p=0$, this measures the vertex expansion of an undirected graph or a directed graph.
	Talagrand studied the case $p=\frac12$ and proved a lower bound on $\varphi_\frac12$ for Boolean hypercubes.
	One can view $\varphi_\frac12$ as a quantity that interpolates between edge conductance and vertex expansion, since it follows from the Cauchy-Schwarz inequality that $\varphi_\frac12(G)^2 \leq \varphi_1(G) \cdot \varphi_0(G)$,
	and Talagrand used his lower bound to recover Margulis' theorem about edge conductance and vertex expansion on hypercubes (see \autoref{sec:related}).

	For an undirected graph $G$, one can use the second smallest eigenvalue $\lambda_2$ of the matrix $I-P_G$ to give upper and lower bound on $\varphi_1(G)$.
	The classical Cheeger's inequality is
	\[
	\lambda_2(I-P_G) \lesssim \varphi_1(G) \lesssim \sqrt{\lambda_2(I-P_G)}.
	\]
	Houdr\'e and Tetali conjectured that the same relations hold for $\varphi_\frac12(G)$ as well when the Markov chain is reversible.
	
	\begin{conjecture}[{\cite{HT04}}] \label{conj:Houdre-Tetali}
	Let $(V,P,\pi)$ be an irreducible and reversible Markov chain.
	Then
	\[
	\lambda_2(I-P) \lesssim \varphi_\frac12(P) \lesssim \sqrt{\lambda_2(I-P)}.
	\]
	\end{conjecture}

 It is clear from the definition that $\varphi_p(G)$ increases as $p$ decreases, and thus $\lambda_2 \lesssim \varphi_1(G) \leq \varphi_p(G)$ for all $p < 1$.
 Therefore, the Houdr\'e-Tetali conjecture is a strengthening of the hard direction of Cheeger's inequality.
	It predicts that when the hard direction of Cheeger's inequality is tight for a graph $G$ such that $\varphi_1(G) \asymp \sqrt{\lambda_2}$, then the graph must satisfy $\varphi_1(G) \asymp \varphi_\frac12(G)$.
	Or, in other words, when $\varphi_1(G) \ll \varphi_\frac12(G)$ such as on the hypercube (see \autoref{rem:hypercubes}) or on the dumbbell graphs, then the hard direction of Cheeger's inequality cannot be tight.
	So, the conjecture can be viewed as an improved Cheeger's inequality in the spirit of~\cite{KLLOT13,KLL16} and this is a main motivation of this work.

	Morris and Peres came close to proving the conjecture with an extra logarithmic factor.
	
	\begin{theorem}[\cite{MP05}] \label{thm:Morris-Peres}
	Let $(V,P,\pi)$ be an irreducible and reversible Markov chain.
	Suppose that $P(v,v) \geq \frac12$ for all $v \in V$.
	Then 
	\[
	\lambda_2(I-P) \gtrsim \frac{\big(\varphi_\frac12(P)\big)^2}{\log\big(1/\varphi_\frac12(P)\big)}.
	\]
	\end{theorem}
	
	Their proof is based on techniques from evolving sets.
	They lower bounded the ``boundary gauge'' $\Psi$ using $\varphi_\frac12$,
	and also upper bounded the mixing rate using $\Psi$ so that they can relate $\lambda_2$ and $\varphi_\frac12$.

	\subsection{Our Results}
	
	We found counterexamples to the Houdr\'e-Tetali conjecture, showing that the extra logarithmic factor is needed.
	
	\begin{theorem} \label{thm:counterexample}
	There are irreducible and reversible Markov chains $(V,P,\pi)$ with
	\[
	\lambda_2(I-P) \lesssim \frac{\log |V|}{|V|^2}
	\quad \textrm{and} \quad
	\varphi_\frac12(P) \gtrsim \frac{\log |V|}{|V|}
	\quad \implies \quad
	\lambda_2(I-P) \lesssim \frac{\big(\varphi_\frac12(P)\big)^2}{\log\big(1/\varphi_\frac12(P)\big)}.
	\]
	\end{theorem}
	
	The counterexample is simple to describe, which is a weighted undirected graph with vertex set $[n]$ and edge weight $P(i,j)$ inversely proportional to $\min\{ |i-j|, n-|i-j|\}^3$.
	See \autoref{sec:counterexample} for details.
	
	On the positive side, we match the result of Morris and Peres using standard spectral arguments.  
	We show that the simple sweep-cut algorithm can be used to output a set $S$ with $\varphi_\frac12(S)$ satisfying the guarantee in \autoref{thm:Morris-Peres}, without the self-loop assumption. 
	See \autoref{sec:Morris-Peres}.
	
	Perhaps more interestingly, the same arguments can be used to prove that the Houdr\'e-Tetali conjecture is true if we replace $\frac12$ by any constant $p > \frac12$.
	
	\begin{theorem} \label{thm:main}
	Let $(V,P,\pi)$ be an irreducible and reversible Markov chain.
	For any $p \in (\frac12,1]$,
	\[
	\big(\varphi_p(P)\big)^2 \leq \frac{4}{2p-1} \cdot \lambda_2(I-P).
	\]
	\end{theorem}
	
	Similar to the discussion after \autoref{conj:Houdre-Tetali},
	this shows that the tight examples of the hard direction of Cheeger's inequality must satisfy $\varphi_{\frac12 + \eps}(G) \lesssim \sqrt{\frac{1}{\eps}} \cdot \varphi_1(G)$ for any $\eps \in (0,\frac12)$.
	Also, this provides an improved analysis of Cheeger's inequality that if $\varphi_1(G) \ll \sqrt{2p-1} \cdot \varphi_p(G)$ then $\varphi_1(G) \ll \sqrt{\lambda_2}$.
	So this result has similar consequences as if Houdr\'e and Tetali's conjecture was true.
	
	Finally, we observe that the same statement as in \autoref{thm:main} can also be proved for non-reversible Markov chains, by replacing $\lambda_2(I-P)$ with the eigenvalue defined for directed graphs by Chung~\cite{Chu05}.
	See \autoref{thm:directed}.

\begin{remark}[$\varphi_p$ for $p < \frac12$]
 For $p < \frac12$, a simple argument shows that an inequality of the form in \autoref{thm:main} cannot hold. To see it, consider the transformation $P \rightarrow (1 - \delta) I + \delta P$ for some parameter $0 < \delta < 1$ (equivalent to adding a large self loop when $\delta$ is small) will scale $\varphi_p$ by a factor $\delta^{p}$, while the second eigenvalue scales by a factor of $\delta$, and so the ratio $\varphi_p(G) / \sqrt{\lambda_2(I - P_G)}$ scales by $\delta^{p - \frac12} \to \infty$ as $\delta \to 0$. When $p = \frac12$, adding self loops does not change the ratio. Thus, it is the first exponent where such an inequality makes sense. 
\end{remark}

	\subsection{Previous Work on Boolean Hypercubes} \label{sec:related}
	
	The isoperimetric constant $\varphi_\frac12$ was initially studied by Talagrand in the Boolean hypercubes.
	Let $\{0,1\}^n$ be the $n$-dimensional hypercube.
	For a point $x \in \{0,1\}^n$, let $x^{\oplus i}$ be the point obtained by flipping the $i$-th bit of $x$.
	For a subset $S \subset \{0,1\}^n$, if $x \notin S$ define $h_S(x)=0$ and otherwise if $x \in S$ define
	\[
	h_S(x) := \big|\{ i \in [n] \mid x^{\oplus i} \notin S \}\big|,
	\]
	so that $\sum_{x} h_S(x)$ is the size of the edge boundary of $S$.
	Let $\mu$ be the uniform distribution on $\{0,1\}^n$ and $\mu(S) := \sum_{x \in S} \mu(x)$.
	The classical Poincar\'e inequality can be stated as, for any $S \subset \{0,1\}^n$,
	\begin{equation} \label{e:Poincare}
	\mu(S) \cdot (1-\mu(S)) \lesssim  \E_{x \sim \mu} \big[h_S(x)\big]. 
	\end{equation}
	Talagrand~\cite{Tal93} proved a strengthening of the Poincar\'e inequality: For any $S \subset \{0,1\}^n$, 
	\begin{equation} \label{e:Talagrand}
	\mu(S) \cdot (1-\mu(S)) \lesssim \E_{x \sim \mu} \big[ {\sqrt{h_S(x)}} \big].
	\end{equation}
	
	The quantity $\E\sqrt{h_S}$ is always smaller than $\E h_S$ and can be seen as a different measure of the boundary information of $S$.
	Let $\partial S := \{x \mid h_S(x) > 0\}$ be the vertex boundary of $S$.
	By the Cauchy-Schwarz inequality, Talagrand's theorem implies Margulis' theorem~\cite{Mar74} that
	\[
	\mu(S)^2 \cdot (1-\mu(S))^2 \lesssim \E_{x \sim \mu} \big[h_S(x)\big] \cdot \mu(\partial S),
	\]
	which was an original motivation for Talagrand to consider the quantity $\E{\sqrt{h_S}}$.
	More recently, both Margulis' and Talagrand's theorems inspired the analogs for directed graphs developed in~\cite{CS16, KMS18}, to make major progresses in analyzing sublinear time algorithms for testing monotone functions.
	See also \cite{EG22,EKLM22} for a proof of a Talagrand's conjecture that further sharpens these inequalities.
	
	The following remark clarifies the connection between $\varphi_p$ and the quantities appearing in Poincar\'e's inequality and Talagrand's inequality.
	\begin{remark}[$\varphi_p$ for Hypercubes] \label{rem:hypercubes}
	For the $n$-dimensional Boolean hypercube $Q_n$, the stationary distribution $\pi$ is simply the uniform distribution $\mu$.
	Note that the numerator in $\varphi_1(Q_n)$ is exactly $\frac{1}{n} \E_{x \in \mu} [h_f(x)]$, and the Poincar\'e inequality translates to $\varphi_1(Q_n) \gtrsim \frac{1}{n}$.
	Similarly, the numerator of $\varphi_{\frac12}(Q_n)$ is exactly $\frac{1}{\sqrt{n}} \E_{x \in \mu} [ {\sqrt{h_f(x)}} ]$, and the Talagrand's inequality translates to $\varphi_{\frac12}(Q_n) \gtrsim \frac{1}{\sqrt{n}}$.
	\end{remark}

	Finally, we note that a parameter similar to $\varphi_p$, called $h_f^p$, was also studied in \cite{EG22}.

	\section{Preliminaries} \label{sec:prelim}

	Given two functions $f, g$, we use $f \lesssim g$ to denote the existence of a positive constant $c > 0$, such that $f \le c \cdot g$ always holds.
	We use $f \asymp g$ to denote $f \lesssim g$ and $g \lesssim f$.
	For positive integers $k$, we use $[k]$ to denote the set $\{1, 2, \dots, k\}$.
	For a function $f: X \rightarrow \R$, $\supp(f)$ denotes the domain subset on which $f$ is nonzero.
	The function $\log x$ refers to the base $e$ logarithm.
	
	{\bf Undirected Graphs:}
	Let $G = (V, E)$ be an undirected graph. 
	Let $w : E \to \R_{\geq 0}$ be a weight function on the edges.
	The weighted degree of a vertex $v$ is defined as $\mathrm{deg}_w(v) := \sum_{e:e \ni v} w(e)$.
	Let $S \subset V$ be a nonempty subset of vertices.
	The edge boundary of $S$ is defined as $\delta(S) := \{e \in E \mid e \cap S \neq \emptyset {\rm~and~} e \cap \overline{S} \neq \emptyset\}$ and $w(\delta(S))$ be the total edge weight of $\delta(S)$.
	The volume of $S$ is defined as $\vol_w(S) := \sum_{v \in S} \mathrm{deg}_w(v)$.
	The edge conductance of $S$ and of $G$ are defined as 
	\[
	\phi(S) := \frac{ w\big(\delta(S)\big) }{\vol_w(S)}
	\quad \textrm{and} \quad
	{\phi}(G) := \min_{S: \vol_w(S) \leq \vol_w(V)/2} {\phi}(S).
	\]
	In an undirected graph, the ordinary random walk has transition matrix $P$ with $P(u,v) = w(uv) / \deg_w(u)$ for every $u,v \in V$.
	If the graph is connected, then the stationary distribution $\pi$ is unique with $\pi(u) = \deg_w(u) / \sum_{v \in V} \deg_w(v)$ for every $u \in V$.
	It is thus straightforward to check that $\phi(S) = \varphi_1(S)$ and $\phi(G) = \varphi_1(G)$, i.e. the two definitions coincide.

	{\bf Directed Graphs:}
	Let $G = (V,E)$ be a directed graph.
	Let $w : E \to \R_{\geq 0}$ be a weight function on the edges.
	The weighted indegree of a vertex $v$ is defined as $d_w^{\rm in}(v) := \sum_{u:\overrightarrow{uv} \in E} w(\overrightarrow{uv})$ and the weighted outdegree of $v$ is defined as $d_w^{\rm out}(v) := \sum_{u:\overrightarrow{vu} \in E} w(\overrightarrow{vu})$.
	In a directed graph, the ordinary random walk has transition matrix $P$ with $P(u,v) = w(\overrightarrow{uv}) / \deg_w^{\rm out}(u)$.
	The stationary distribution $\pi$ has no easy description but is unique as long as the directed graph is strongly connected.
	There are different notions of directed edge conductance for directed graphs.
	In analyzing random walks, the standard definition is exactly $\varphi_1(G)$ as described in \autoref{def:iso}, and this quantity is closely related to the mixing time of random walks; see e.g.~\cite{LP17,Chu05,MP05}.
	In analyzing graph partitioning, there is a definition that extends the edge conductance above to directed graphs, which will not be used in this paper; see e.g.~\cite{Yos16,LTW23}.

	{\bf Spectral Graph Theory:}
	Given an undirected graph $G = (V, E)$ with a weight function $w : E \to \R_{\geq 0}$, its adjacency matrix $A = A(G)$ is an $|V| \times |V|$ matrix where the $(u, v)$-th entry is $w(uv)$. 
	The Laplacian matrix is defined as $L := D - A$, where $D:=\diag(\{\deg_w(v)\}_{v \in V})$ is the diagonal degree matrix. 
	The normalized adjacency matrix is defined as $\mathcal{A} := D^{-1/2} A D^{-1/2}$, and the normalized Laplacian matrix is defined as $\mathcal{L} := I - \mathcal{A}$. 
	Let $\lambda_1(\L) \leq \lambda_2(\L) \leq \cdots \leq \lambda_n(\L)$ be the eigenvalues of $\L$.
	It is known that $\lambda_1(\L)=0$ with eigenvector $D^{1/2} \vec{1}$, and
	\begin{equation*} \label{e:lambda2}
	\lambda_2(\L) = \min_{g \perp D^{1/2} \vec{1}} \frac{g^T \L g}{g^Tg}
	= \min_{f \perp D \vec{1}} \frac{f^T L f}{f^T D f}
	= \min_{f \perp D \vec{1}} \frac{\sum_{uv \in E} w(uv) \cdot (f(u)-f(v))^2}{\sum_v \deg_w(v) \cdot f(v)^2}.
	\end{equation*}
	Cheeger's inequality~\cite{Che70,AM85,Alo86} is a fundamental result in spectral graph theory that connects edge conductance of an undirected graph $G=(V,E)$ to the second smallest eigenvalue of its normalized Laplacian matrix:
	\begin{equation*} \label{e:Cheeger}
	\frac{\lambda_2}{2} \leq \phi(G) \leq \sqrt{2\lambda_2}.
	\end{equation*}
	
	The random walk transition matrix $P$ is similar to the normalized Laplacian matrix $\mathcal{A}$, and the matrix $I-P$ is similar to the normalized Laplacian matrix $\mathcal{L}$. 
	In particular, $I-P$ enjoys the same spectral properties as $\mathcal{L}$ with real eigenvalues and a quadratic form characterization of $\lambda_2$ as above; see \autoref{lem:lambda2}.
	
	Chung~\cite{Chu05} defined the Laplacian matrix of a directed graph and used it to prove an analog of Cheeger's inequality.
	These will be stated in \autoref{sec:Chung}.

	\section{Positive Results} \label{sec:main-results}
	
	To prove \autoref{thm:main}, we follow standard spectral arguments used in proving Cheeger-type inequalities, in Trevisan's style.  
	First, we start with the second eigenvector $f_2 : V \to \R$ and truncate it so that the $\pi$ weight of its support is at most half while preserving its Rayleigh quotient.
	The proof of the following lemma is standard and we defer it to the end of this section.
	
	\begin{lemma} \label{lem:lambda2}
	Let $(V,P,\pi)$ be an irreducible and reversible Markov chain.
	Let $f_2$ be an eigenvector associated to the second smallest eigenvalue of the matrix $I-P$, with $\pi(\{ v \mid f_2(v) > 0\} ) \leq \frac12$.
	Define the truncated vector $f$ such that $f(v) := \max\{f_2(v),0\}$ for all $v \in V$.
	Then
	\[
	\lambda_2(I-P) \geq \frac{\sum_{f(i) \geq f(j)} \pi(i) \cdot P(i,j) \cdot (f(i)-f(j))^2}{\sum_{i \in V} \pi(i) f(i)^2}.
	\]
	\end{lemma}
	
	Then the plan is to prove that one of the level sets has small isoperimetric constant.
	We index the vertices by $[n]$ and order them so that $f(i) \leq f(j)$ for $i \leq j$.
	For any $t\geq 0$, define the level set $S_t := \{ i \in [n] \mid f(i)^2 > t \}$.
	By the construction of $f$, it holds that $\pi(S_t) \leq \frac12$ for any $t \geq 0$, and so $\varphi_p(P) \leq \min_{t : t \geq 0} \varphi_p(S_t)$.
	For convenience, we rescale $f$ so that $\max_i f(i) = 1$.
	
	To prove that one of $S_t$ has small isoperimetric constant,
	we choose a uniform random $t \in [0,1]$ and consider $\varphi_p(S_t)$.
	We will bound the expected value of the numerator of $\varphi_p(S_t)$ and of the denominator of $\varphi_p(S_t)$ and conclude that there exists a $t$ with $\varphi_p(S_t)$ at most the ratio of the expected values, i.e. $\displaystyle \min_{t:t\geq 0} \varphi_p(S_t) \leq \E_t[\textrm{numerator~of~}\varphi_p(S_t)]~/~\E_t[\textrm{denominator~of~}\varphi_p(S_t)]$.
	
	The expected value of the denominator is easy to analyze.
	Since we choose $t$ uniformly randomly, each vertex $i$ is included in $S_t$ with probability $f(i)^2$, and thus
	\[
	\E_t[\pi(S_t)] = \sum_{i \in V} \pi(i) \cdot f(i)^2.
	\]
	
	The rest of the proof is to analyze the expected value of the numerator $\sum_{i \in V} \pi(i) \cdot P(i,\overline{S_t})^p$.
	For a vertex $i$, if the random threshold $t$ is between $f(j)^2$ and $f(j-1)^2$ with $f(j) > f(j-1)$, then $P(i,\overline{S_t}) = P(i,[j-1])$,
	and so
	\begin{eqnarray*}
	\E_t\big[P(i,\overline{S_t})^p\big]
	& = & \sum_{j=1}^i \big( f(j)^2 - f(j-1)^2 \big) \cdot P(i,[j-1])^p
	\\
	& = & \sum_{j=1}^i \big( f(j)^2 - f(j-1)^2 \big) \cdot \sum_{l=1}^{j-1} \big( P(i,[l])^p - P(i,[l-1])^p\big)
	\\
	& = & \sum_{l=1}^{i-1} \big( f(i)^2 - f(l)^2 \big) \cdot \big( P(i,[l])^p - P(i,[l-1])^p\big), 
	\end{eqnarray*}
	where the second equality is by writing a telescoping sum and the third equality is by a change of summation.
	So, the expected numerator $\E_t[\sum_{i=1}^n \pi(i) \cdot P(i,\overline{S_t})^p]$ is
	\begin{eqnarray*}
	& & \sum_{i=1}^n \sum_{j=1}^{i-1} \pi(i) \cdot \big( f(i)^2 - f(j)^2 \big) \cdot \big( P(i,[j])^p - P(i,[j-1])^p\big) 
	\\
	& \leq &
	\sqrt{ \sum_{i=1}^n \sum_{j=1}^{i-1} \pi(i) \cdot (f(i) - f(j))^2 \cdot P(i,j)} \\ 
	& & \quad \quad \cdot
	\sqrt{ \sum_{i=1}^n \sum_{j=1}^{i-1} \pi(i) \cdot (f(i) + f(j))^2 \cdot \frac{ \big( P(i,[j])^p - P(i,[j-1])^p\big)^2 }{P(i,j)} }
	\\
	& \leq &
	\sqrt{ \lambda_2 \cdot \sum_{i=1}^n \pi(i) \cdot f(i)^2} \cdot 
	\sqrt{ \sum_{i=1}^n 4 \cdot \pi(i) \cdot f(i)^2 \cdot \sum_{j=1}^{i-1} \frac{ \big( P(i,[j])^p - P(i,[j-1])^p\big)^2 }{P(i,[j]) - P(i,[j-1])} }, 
	\end{eqnarray*}
	where the first inequality is by Cauchy-Schwarz, and the second inequality is by \autoref{lem:lambda2} and $(f(i)+f(j))^2 \leq 4f(i)^2$ and $P(i,j) = P(i,[j]) - P(i,[j-1])$.
	
	To upper bound the inner sum of the second term,
	we denote $a_j := P(i,[j])$ and it suffices to upper bound the sum of the form
	$\sum_{j=1}^{n} (a_j^p - a_{j-1}^p)^2 / (a_j - a_{j-1})$ with $a_0=0$ and $a_n \leq 1$, with a bound independent of $n$.
	Let $C(n,a)$ denote the supremum of the sum when $a_n=a$.
	Note that $C(n,a) = a^{2p-1} \cdot C(n,1)$ by a simple scaling argument.
	Let $(a_{i})_{i=1}^n$ be an optimal sequence that achieves the supremum of $C(n,1)$.
	Then,
	\begin{eqnarray*}
	C(n,1)
	~=~ C(n-1,a_{n-1}) + \frac{(1-a_{n-1}^p)^2}{1-a_{n-1}} 
	& = & a_{n-1}^{2p-1} \cdot C(n-1,1) + \frac{(1-a_{n-1}^p)^2}{1-a_{n-1}} \\
	& \leq & a_{n-1}^{2p-1} \cdot C(n,1) + \frac{(1-a_{n-1}^p)^2}{1-a_{n-1}}.
	\end{eqnarray*}
	It follows that
	\[
	C(n,1) 
	\leq \sup_{a \in [0,1]} \frac{(1 - a^p)^2}{(1 - a)(1 - a^{2p-1})}
	\leq \sup_{a \in [0,1]} \frac{(1 - a^p)^2}{(2p-1)(1 - a)^2}
	\leq \frac{1}{2p-1},
	\]
	where the second inequality is by the mean value theorem that $1-a^{2p-1} \geq (2p-1)(1-a)$ and the last inequality is because $a \in [0,1]$ and $p \in (\frac12,1]$.
	Clearly, $C(n,1) \geq C(n,a_n)$ for any $a_n \in [0,1]$, and so the inner sum of the second term in the expected numerator is at most $\frac{1}{2p-1}$.
	Putting together, this completes the proof of \autoref{thm:main} as
	\[
	\varphi_p(P) 
	\leq \min_{t:t>0} \varphi_p(S_t) 
	\leq \frac{ \E_t[\sum_{i=1}^n \pi(i) \cdot P(i,\overline{S_t})^p] }{ \E_t[\pi(S_t)] }
	\leq 2 \sqrt{\lambda_2} \sqrt{\frac{1}{2p-1}},
	\]
	which implies that $\big(\varphi_p(P)\big)^2 \leq \frac{4}{2p-1} \cdot \lambda_2.$
	
	\subsection{Recovering Morris and Peres's Result} \label{sec:Morris-Peres}
	
	To recover \autoref{thm:Morris-Peres},
	we follow the same arguments but add a truncation step so that the sequence $a_i$ above will start with $a_0 \approx \varphi_{\frac12}(P)$.
	In this subsection, we plug in $p=\frac12$.
	As above, the main work is to upper bound the expected value of the numerator.
	Recall that
	\[
	\E_t\Big[\sqrt{P(i,\overline{S_t})} \Big]
	= \sum_{j=1}^{i-1} \big( f(i)^2 - f(j)^2 \big) \cdot \big( \sqrt{P(i,[j])} - \sqrt{P(i,[j-1])}\big), 
	\]
	Let $l_i$ be the index such that $\sqrt{P(i,[l_i])} \leq \frac12 \varphi_\frac12(P)$ but $\sqrt{P(i,[l_i+1])} > \frac12 \varphi_\frac12(P)$.
	Then, we can upper bound the right hand side by
	\begin{eqnarray*}
	\E_t\Big[\sqrt{P(i,\overline{S_t})} \Big]
	\leq \frac12 \varphi_\frac12(P) \cdot f(i)^2 & + & (f(i)^2 - f(l_i + 1)^2) \Big( \sqrt{P(i, [l_i + 1])} - \frac{1}{2} \varphi_\frac12(P) \Big) \\
	& + & \sum_{j = l_i + 2}^{i - 1} (f(i)^2 - f(j)^2) \left( \sqrt{P(i, [j])} - \sqrt{P(i, [j-1])} \right). 
	\end{eqnarray*}
	To shorten the expression, let us use the notations 
	\[a_{i,0} = \frac12 \varphi_\frac12(P)
	\quad \textrm{and} \quad
	a_{i,j} = \sqrt{P(i, [l_i + j])}.
	\]
	Summing over $i$ and using these notations, the expected numerator is
	\begin{eqnarray*}
	\E_t\Big[\sum_{i=1}^n \pi(i) \cdot \sqrt{P(i,\overline{S_t})}\Big]
	 & \leq & \frac12 \varphi_\frac12(P) \cdot \sum_{i=1}^n \pi(i) \cdot f(i)^2 \\
	 & & \quad + \underbrace{\sum_{i=1}^n \pi(i) \cdot \sum_{j=1}^{i-l_i-1} (f(i)^2 - f(l_i+j)^2) \cdot(a_{i,j} - a_{i,j-1})}_{(*)}
	\end{eqnarray*}
	Applying Cauchy-Schwarz as before gives
	\begin{eqnarray*}
	(*) 
	& \leq & 
	\sqrt{ \sum_{i=1}^n \sum_{j=1}^{i-l_i-1} \pi(i) \cdot (f(i) - f(l_i+j))^2 \cdot P(i,l_i+j)~} \\
	& & \quad \cdot 
	\sqrt{ \sum_{i=1}^n \sum_{j=1}^{i-l_i-1} \pi(i) \cdot (f(i) + f(l_i+j))^2 \cdot \frac{ \big( a_{i,j} - a_{i,j-1} \big)^2 }{P(i,l_i+j)} }
	\\
	& \leq &
	\sqrt{ \lambda_2 \cdot \sum_{i=1}^n \pi(i) \cdot f(i)^2} \cdot 
	\sqrt{ \sum_{i=1}^n 4 \cdot \pi(i) \cdot f(i)^2 \cdot \underbrace{\sum_{j=1}^{i-1} \frac{ a_{i,j} - a_{i,j-1} }{ a_{i,j} + a_{i,j-1}}}_{(**)} }, 
	\end{eqnarray*}
	where the second inequality uses \autoref{lem:lambda2} and $P(i,l_i+j) = a_{i,j}^2 - a_{i,j-1}^2$.
	
	To upper bound $(**)$, we let $b_i := a_{i,j}$ and use that $\frac12 \varphi_\frac12(P) = b_0 \leq b_1 \leq \ldots \leq b_m \leq 1 =: b_{m + 1}$ to upper bound the function 
	\[
	f : (b_{0}, b_{1}, \ldots, b_{m}) \to 
	\frac{b_{1}-b_{0}}{b_{1}+b_{0}} 
	+ \frac{b_{2}-b_{1}}{b_{2}+b_{1}} 
	+ \cdots
	+ \frac{b_{m}-b_{m-1}}{b_{m}+b_{m-1}}  
	+ \frac{1-b_{m}}{1+b_{m}}  
	\]
	The partial derivative of $f$ is
	\[
	\frac{\partial f}{\partial b_i}  
	= \frac{ 2 b_{i-1}}{(b_{i} + b_{i-1})^2} - \dfrac{2 b_{i+1}}{(b_{i+1} + b_{i})^2}
	= \frac{2 (b_{i+1} - b_{i-1})(b_{i-1}b_{i+1} - b_i^2)}{(b_{i} + b_{i-1})^2 (b_{i+1} + b_i)^2}.
	\]
	Since $b_{i+1} - b_{i-1} > 0$ by definition,
	the function increases up until $b_i^2 = b_{i-1} b_{i+1}$ and then decreases.
	So, the maximum is attained when $b_i = (b_{0})^\frac{m+1-i}{m+1}$ with $b_0 = \frac12 \varphi_\frac12(P)$ and $b_{m+1}=1$, in which case the sum is
	\[
	\sum_{i=1}^{m+1} \frac{b_i - b_{i-1}}{b_i + b_{i-1}}
	= \sum_{i=1}^{m+1} \frac{1 - \frac{b_{i-1}}{b_i}}{1 + \frac{b_{i-1}}{b_i}}
	= \sum_{i=1}^{m+1} \frac{1 - b_0^{\frac1{m+1}} }{1 + b_0^\frac1{m+1} }
	= (m+1) \cdot \frac{1 - b_0^{\frac1{m+1}} }{1 + b_0^\frac1{m+1} }
	\]
	For $b_0 \in [0,1]$, this value is increasing when $m$ increases, and so the sum is upper bounded by
	\begin{eqnarray*}
	(**) 
	\leq \lim_{x \to \infty} x \cdot \frac{1-b_0^{\frac{1}{x}}}{1+b_0^{\frac{1}{x}}} 
	= \lim_{x \to \infty} x \cdot \frac{1-b_0^{\frac{1}{x}}}{2} 
	& = & \lim_{y \to 0} \frac{1-b_0^y}{2y} \\
	& = & \lim_{y \to 0} \frac{-b_0^y \log b_0}{2}
	= \frac12 \log \frac{1}{b_0}
	= \frac12 \log \frac{2}{\varphi_\frac12(P)}, 
	\end{eqnarray*}
	where the third last equality is by L'H\^{o}pital's rule. Plugging this back into $(**)$ and $(*)$,
	the expected numerator is
	\begin{eqnarray*}
	\E_t\Big[\sum_{i=1}^n \pi(i) \cdot \sqrt{P(i,\overline{S_t})}\Big]
	& \leq & \frac12 \varphi_\frac12(P) \cdot \sum_{i=1}^n \pi(i) \cdot f(i)^2 \\
	& & \quad + \sqrt{ \lambda_2 \cdot \sum_{i=1}^n \pi(i) \cdot f(i)^2}~ \cdot 
	\sqrt{ \sum_{i=1}^n 2 \log \frac{2}{\varphi_\frac12(P)} \cdot \pi(i) \cdot f(i)^2 }.
	\end{eqnarray*}
	As before, the expected denominator is $\E_t[\pi(S_t)] = \sum_{i \in V} \pi(i) \cdot f(i)^2$.
	Putting together, 
	\begin{eqnarray*}
	\varphi_\frac12(P) 
	\leq \min_{t:t>0} \varphi_\frac12(S_t) 
	\leq \frac{ \E_t\Big[\sum_{i=1}^n \pi(i) \cdot \sqrt{P(i,\overline{S_t})}\Big] }{ \E_t[\pi(S_t)] }
	\leq \frac12 \varphi_\frac12(P) + \sqrt{2 \log \frac{2}{\varphi_\frac12(P)} \lambda_2}.
	\end{eqnarray*}
	Rearranging recovers \autoref{thm:Morris-Peres}.

	\subsection{Non-Reversible Markov Chains} \label{sec:Chung}
	
	Houdr\'e and Tetali only formulated \autoref{conj:Houdre-Tetali} for reversible Markov chains of which the eigevalues of $I-P$ are real.
	For non-reversible Markov chains, we observe that Chung's definition of eigenvalues for directed graphs~\cite{Chu05} can be used to obtain the same results in \autoref{thm:Morris-Peres} and \autoref{thm:main}.
	
	Given a directed graph $G = (V, E)$ with a weight function $w : E \to \R_{\geq 0}$, let $P_G$ be the transition matrix of the ordinary random walks on $G$ with $P_G(u,v) = w(uv) / \sum_{v \in V} w(uv)$ for each edge $uv \in E$. 
	Suppose $G$ is strongly connected, then there is a unique stationary distribution $\pi : V \to \R_+$ such that $\pi^T P = \pi^T$.
	Let $\Pi = \diag(\pi)$.
	Chung defined the Laplacian of the directed graph $G$ as 
	\[
	\vec{L}_G : = I - \frac12 \Big( \Pi^{\frac12}P \Pi^{-\frac12} + \Pi^{-\frac12} P^T \Pi^{\frac12} \Big).
	\]
	Since $\vec{L}_G$ is a real symmetric matrix, its eigenvalues are real.
	Let $\lambda_2$ be the second smallest eigenvalue of $\vec{L}_G$.
	Chung~\cite{Chu05} proved an analog of Cheeger's inequality that
	\[
	\frac12 \varphi_1(G)^2 \leq  \lambda_2(\vec{L}_G) \leq 2 \varphi_1(G).
	\]
	We observe that $\lambda_2(\vec{L}_G)$ can be used to extend our results to non-reversible Markov chains.
	
	\begin{theorem} \label{thm:directed}
	Let $(V,P,\pi)$ be an irreducible Markov chain.
	For any $p \in (\frac12,1]$,
	\[
	\big(\varphi_p(P)\big)^2 \leq \frac{4}{2p-1} \cdot \lambda_2(\vec{L}_G).
	\]
	For $p = 1/2$,
	\[
	\lambda_2(\vec{L}_G) \gtrsim \frac{\big(\varphi_\frac12(P)\big)^2}{\log\big(1/\varphi_\frac12(P)\big)}.
	\]
	\end{theorem}

	Note that the main proofs of \autoref{thm:main} and \autoref{thm:Morris-Peres} (i.e.~computing the expected numerator) did not require the Markov chain to be reversible.  
	The reversible assumption was only used in characterizing the second eigenvalue in \autoref{lem:lambda2}.
	The following is an analog of \autoref{lem:lambda2} for non-reversible Markov chains using Chung's definition of the second eigenvalue of directed graphs.

	\begin{lemma} \label{lem:lambda2-directed}
	Let $(V,P,\pi)$ be an irreducible Markov chain.
	Let $v_2$ be an eigenvector associated to the second smallest eigenvalue of the matrix $\vec{L}_G$.
	Define the reweighted eigenvector $f_2 := \Pi^{-\frac12} v_2$, 
	with $\pi(\{v: f_2(v) \geq 0\}) \leq \frac12$.
	Define the truncated vector $f: = \max(f_2, 0)$. 
	Then
	\[
	\lambda_2(\vec{L}_G) \geq \frac{\sum_{u,v \in V: f(u) \geq f(v)} \pi(u) \cdot P(u, v) \cdot (f(u) - f(v))^2 }{\sum_{v \in V} \pi (v) f(v)^2}.
	\]
	\end{lemma}
	
	With this lemma, we can follow the proofs of \autoref{thm:main} and \autoref{thm:Morris-Peres} verbatim as in above, by defining level sets $S_t$ using the truncated vector $f$ and computing the expected numerator and denominator and so on.
	
	This concludes the proof of \autoref{thm:directed}.
	We will prove \autoref{lem:lambda2-directed} in the next subsection.

	\subsection{Proofs of Auxiliary Lemmas}

	In this subsection, we prove \autoref{lem:lambda2} and \autoref{lem:lambda2-directed}.
	The proofs are standard but we include them for completeness.
	
	\begin{proof}[Proof of \autoref{lem:lambda2}]
	For $f, g: V \to \mathbb{R}$, we define $\langle f,g \rangle_\pi := \sum_{i \in V} \pi(i) \cdot f(i) \cdot g(i)$. 
	By definition of the second eigenvector,
	$ \langle (I-P) f_2,f_2 \rangle_\pi = \lambda_2 \langle f_2,f_2\rangle_\pi$.
	
	For $f := \max(f_2, 0)$, note that $Pf\geq Pf_2$, as 
	\[(Pf)(i) = \sum_{j \in V} p(i,j) f(j) \geq \sum_{j \in V} p(i,j) f_2 (j) = (Pf_2)(i),\] and thus
	\[
	\langle P f,f \rangle_\pi 
	= \sum_{i \in V : f_2(i) \geq 0} \pi(i) \cdot (Pf)(i) \cdot f(i) 
	\geq \sum_{i \in V, f_2(i) \geq 0} \pi(i) \cdot (Pf_2)(i) \cdot f_2(i)
	= (1-\lambda_2) \langle f,f \rangle_\pi,
	\]
	where the last equality uses that $Pf_2 = (1-\lambda_2)f_2$.
	It follows that
	\[
	\lambda_2 \geq \frac{\inner{(I - P) f}{f}_\pi}{\inner{f}{f}_\pi}.
	\]
	The denominator is the same as the denominator in the statement. 
	It remains to check that the numerator is also the same as the numerator in the statement.
	By direct calculation,
	\begin{align*}
		\inner{(I - P) f}{f}_\pi
		& =  \sum_{i \in V}\pi (i) f(i)^2 - \sum_{i \in V} \pi(i) \sum_{j \in V} P(i,j) f(j) f(i) \\
		& = \sum_{i \in V} \sum_{j \in V} \pi(i) P(i,j) f(i)^2 - \sum_{i \in V} \sum_{j \in V}  \pi(i)  P(j,i) f(j) f(i) \\
		& = \sum_{i \in V} \sum_{j \in V} \pi (i) P(i,j)\Big( \frac{1}{2} (f(i)^2 + f(j)^2) - f(i)f(j)\Big) \\
		& = \sum_{i > j} \pi(i) P(i,j) (f(i) - f(j))^2,
	\end{align*}
	where the second equality uses $\sum_{j \in V} P(i,j) = 1$ and the third equality uses reversibility which gives $\pi(i)P(i,j) = \pi(j)P(j,i)$ for all $i,j \in V$, to get $\sum_{i,j} \pi(i) P(i,j) f(i)^2 = \sum_{i,j} \pi(i) P(i,j) f(j)^2$.
	\end{proof}

	To prove \autoref{lem:lambda2-directed}, we will use the following facts about $\lambda_2(\vec{L}_G)$ in \cite{Chu05}.
	\begin{lemma}[\cite{Chu05}] 
	Let $G=(V,E)$ be a strongly connected directed graph and $\pi$ be its stationary distribution.
	The second smallest eigenvalue $\lambda_2$ of the directed Laplacian $\vec{L}_G$ satisfies 
	\[
	\lambda_2 
	= \inf_{f \perp \pi}  \frac{\sum_{u,v \in V} \pi(u) \cdot P(u,v) \cdot |f(u) - f(v)|^2}{\sum_{v \in V} \pi (v) \cdot |f(v)|^2}
	\]
	Suppose $v_2$ is an eigenvector of $\vec{L}_G$ associated with eigenvalue $\lambda_2$. Then, for the reweighted eigenvector $f_2 := \Pi^{-\frac12} v_2$, for all $u \in V$, 
	\[
	\lambda_2 \cdot f_2(u) \cdot \pi(u) = \frac{1}{2} \sum_{v} \big(f_2(u) - f_2(v)\big) \cdot  \big(\pi(u) P(u,v) + \pi(v) P(v,u)\big).
	\]
	\end{lemma}

	\begin{proof}[Proof of \autoref{lem:lambda2-directed}]
	We claim that the truncated vector $f := \max\{f_2,0\}$ satisfies
	\[
	\lambda_2 \cdot f(u) \cdot \pi(u) \geq \frac{1}{2} \sum_{v} \big(f(u) - f(v)\big) \cdot  \big(\pi(u) P(u,v) + \pi(v) P(v,u)\big).
	\]
	for all $u \in V$.
	Indeed, for $u$ such that $f(u) > 0$,
	\begin{eqnarray*} 
		\lambda_2 \cdot f(u) \cdot \pi(u) 
		& = & \lambda_2 \cdot f_2(u) \cdot \pi(u)
		\\
		& = & \frac{1}{2} \sum_{v \in V} \big(f_2(u) - f_2(v)\big) \cdot \big(\pi(u) P(u,v) + \pi(v) P(v,u)\big) 
		\\
		& \geq & \frac{1}{2} \sum_{v \in V} \big(f(u) - f(v)\big) \cdot \big(\pi(u) P(u,v) + \pi(v) P(v,u)\big),
	\end{eqnarray*}
	where the second equality is by the fact above and the last inequality is by $f_2(u) - f_2(v) \geq f(u) - f(v)$ for all $u,v \in V$ due to truncation.
	For $u$ such that $f(u) = 0$, the inequality holds trivially because
	\[
	\lambda_2 \cdot f(u) \cdot \pi(u) 
	= 0
	\geq \frac{1}{2} \sum_{v} \big(- f(v)\big) \cdot \big(\pi(u) P(u,v) + \pi(v) P(v,u)\big)
	\]
	as $f(v) \geq 0$ for all $v$ by truncation. 
	Thus the claim follows.
	Multiplying both sides of the claim by $f(u)$ and then summing over all $u$ gives
	\begin{eqnarray*}
		\lambda_2 \cdot \sum_{u \in V} f^2(u) \pi(u)  
		& \geq & \frac12 \sum_{u \in V} f(u) \sum_{v \in V} \big(f(u) - f(v)\big) \cdot \big(\pi(u) P(u,v) + \pi(v) P(v,u)\big) 
		\\
		& = & \frac{1}{2}\sum_{u \in V} \sum_{v \in V} \pi(u) \cdot P(u,v) \cdot \Big(\frac{1}{2} f(u)^2 + \frac{1}{2} f(v)^2 - f(u) f(v) \Big) 
		\\
		& = & \frac{1}{2} \sum_{u \in V} \sum_{v \in V} \pi(u) \cdot P(u,v) \cdot \big(f(u) - f(v)\big)^2.
	\end{eqnarray*}
	This is equivalent to the statement where the sum is over pairs with $f(u) \geq f(v)$.
	\end{proof}

	\section{Counterexamples} \label{sec:counterexample}
	
	In this section, we prove \autoref{thm:counterexample} by constructing a family of counterexamples and bounding their second eigenvalues and $\varphi_\frac12$ value.
	The construction is simple.
	
	\begin{definition}[Counterexamples] \label{def:counterexample}
	Let $G_n$ be a graph with vertex set $[n]$.
	For each $i, j \in [n], i \neq j$, the edge weight is
	\[
	P(i,j) = \frac{1}{C \big( \min\{ |i-j|, n-|i-j|\} \big)^3 },
	\]
	where $C = \sum_{i=1}^n 1/\min\{ |i-j|, n-|i-j|\}^3 $ is the normalizing constant to make the graph $1$-regular.
	\end{definition}
	
	We will prove the two claims in \autoref{thm:counterexample} about the second smallest eigenvalue and the $\varphi_{\frac1/2}(G)$ value.
	First, we analyze the second smallest eigenvalue, based on the construction that $I-P$ is a circulant matrix.
	
	\begin{lemma}
	For $G_n$ in \autoref{def:counterexample}, the second smallest eigenvalue of $I-P$ is
	\[
	\lambda_2(I-P) \lesssim \frac{\log n}{n^2}.
	\] 
	\end{lemma}
	
	\begin{proof}
	By our construction, the graph $G_n$ is cyclic that $P(i,j) = P( (i+k) \mathrm{~mod~} n, (j+k) \mathrm{~mod~} n)$ for all $i, j, k \in [n]$.
	So the matrix $I-P$ is a circulant matrix of the form
	\[ I-P = \begin{pmatrix} a_0 & a_1 & a_2 & \dots & a_{n - 1} \\ a_{n-1} & a_0 & \cdots & \cdots & a_{n -2} \\ a_{n-2} & a_{n - 1} & \ddots & \ddots & a_{n - 3} \\ \vdots & \ddots & \ddots & \ddots & \vdots \\ a_{1} & a_2 & \dots & \dots & a_{0} \end{pmatrix}  \]
	where $a_0 = 1$ and $a_j = -P(1,j+1)$ for all $j \in [n]$.
	It is well-known that an $n \times n$ circulant matrix with first row entries $a \in \mathbb{R}^n$ has eigenvalues and corresponding eigenvectors
	\[
	\Big\{\sum_{i = 0}^{n -1} a_{i} {\omega_k}^{i} \Big\}_{k = 0}^{n -1}
	\quad \mathrm{and} \quad
	\left\{ \left( 1, \omega_k, \ldots, \omega_k^{n-1} \right)^T \right\}_{k=0}^{n-1}
	\]
	where $\omega_k := e^{\frac{2\pi k \imath}{n}}$ are the $n$-th roots of unity for $k \in [n] $ (where $\imath$ denotes the imaginary number).
	
	So, the second smallest eigenvalue $\lambda_2$ of $I-P$ corresponds to the first $n$-th root of unity $\omega := \omega_1 = e^{\frac{2 \pi \imath}{n}}$, and 
	\[
	\lambda_2 
	= \sum_{i = 0}^{n -1} a_i \cdot {\omega}^{i} 
	= \sum_{i = 2}^n P(1, i) - \sum_{i = 2}^{\lfloor n/2 \rfloor + 1} P(1, i) \cdot {\omega}^{i-1} - \sum_{i =\lfloor n/2 \rfloor + 2}^{n}  P(1, i) \cdot {\omega}^{i - 1}.
	\]
	We consider two cases, when $n$ is odd and $n$ is even.
	When $n = 2k+1$ is odd, note that by definition $P(1,i) = P(1,2k+3-i)$ for $2 \leq i \leq k+1$ and so we can pair up the terms in the above equation to get
	\[
	\lambda_2 = \sum_{i = 2}^{k + 1} P(1, i) \Big( 2 - {\omega}^{i-1} - \frac{1}{\omega^{i - 1}} \Big)
	= - \sum_{i = 2}^{k + 1} P(1, i) \bigg({\omega}^{\frac{i-1}{2}} - \frac{1}{\omega^{\frac{i - 1}{2}}} \bigg)^2.
	\]
	Using the definition of $\omega^k := e^{\frac{2 \pi k \imath}{n}} = \cos \frac{2 k\pi}{n} + \imath \sin \frac{2 k\pi}{n}$, 
	it follows that 
	\begin{eqnarray*}
	\lambda_2 
	= - \sum_{i = 2}^{k + 1} P(1, i) \Big({\omega}^{\frac{i-1}{2}} - \overline{\omega}^{\frac{i - 1}{2}} \Big)^2
	& = - \sum_{i = 2}^{k + 1} P(1, i) \left( 2 \imath \sin \dfrac{(i - 1)\pi}{n} \right)^2 \\
	& = 4 \sum_{i = 2}^{k + 1} P(1, i) \left(\sin \dfrac{(i - 1)\pi}{n} \right)^2.
	\end{eqnarray*}
	Finally we use the fact that $\sin x < x$ and that $P(1, i) =  \frac{1}{ C \cdot \min\{|i - 1|, n - |i-1|\}^3} < \frac{1}{(i - 1)^3}$ for $2 \leq i \leq k + 1$ as $C \geq 1$ to conclude that
	\[
	\lambda_2 
	< 4 \sum_{i = 2}^{k + 1} \frac{1}{(i - 1)^3} \Big( \frac{(i - 1)\pi}{n} \Big)^2
	= 4 \sum_{i = 1}^{k} \frac{1}{i} \Big( \frac{\pi^2}{n^2} \Big)
	\lesssim \frac{\log n}{n^2}.
	\]
	When $n = 2k$ is even, the proof follows along the same lines, but we need to remove the term $k = \frac{n}{2} + 1$ in the sum because $\omega^{n / 2} = - 1$. However, it only contributes a term of $O\left(\frac{1}{n^3} \right)$ to the sum, which is negligible.
	\end{proof}

	It remains to prove that $\varphi_\frac12(P) \gtrsim \frac{\log n}{n}$.
	As this graph is symmetric, our intuition is that $\varphi_\frac12(P)$ attains its minimum at the set $S = \{1,\ldots,\frac{n}{2}\}$.
	In this case, for each vertex $i \in S$, 
	\[
	\sqrt{P(i,\overline{S})} \geq \sqrt{\sum_{j = i}^{i + \frac{n}{2}} \frac{1}{j^3}} \approx \frac{1}{i} - \frac{1}{i + \frac{n}{2}} \geq \frac{1}{2i},
	\]
	which implies that
	\[
	\varphi_\frac12(S) \gtrsim \sum_{i=1}^{\frac{n}{2}} \frac1n \sqrt{P(i,\overline{S})} \gtrsim \frac{1}{n} \sum_{i=1}^{\frac{n}{2}} \frac{1}{i} \gtrsim \frac{\log n}{n}.
	\]
	Our plan was to prove that $S$ indeed attains the minimum, but we do not have such a proof.
	Instead, we will work on a slightly different lower bound, which satisfies a concavity property that allows us to argue that sets of consecutive vertices attain the minimum, in order to prove the lower bound.
	It turns out that the proof is a bit long and we will present it in the next subsection.

	\subsection{Proof of \texorpdfstring{$\varphi_{\frac{1}{2}}$}{} Lower Bound}
	
	First, we set up some notations for the proof.
	Let us partition the vertex set of $G_n$ into two sets $A$ and $B := G \setminus A$ with $|A| \leq |B|$. 
	As the graph $G_n$ is cyclic, we can arrange the vertices $V = [n]$ in a clockwise manner and without loss of generality we assume $1 \in A$ and $n \in B$.
	Let us divide the vertices of $A$ and $B$ into contiguous sets $A_1,B_1,A_2,B_2, \ldots, A_k,B_k$ in the cyclic representation, and denote their sizes by $a_i := |A_i|$ and $b_i := |B_i|$ for $1 \leq i \leq k$.
	More explicitly, for $1 \leq i \leq k$, the vertices in $A_i$ and $B_i$ are
	\[
	A_i = \Big\{\sum_{j=1}^{i-1} a_j + \sum_{j=1}^{i-1} b_j + 1, \ldots, \sum_{j=1}^i a_j + \sum_{j=1}^{i-1} b_j \Big\} 
	\mathrm{,}
	B_i = \Big\{\sum_{j=1}^{i} a_j + \sum_{j=1}^{i-1} b_j + 1, \ldots, \sum_{j=1}^i a_i + \sum_{j=1}^{i-1} b_i \Big\}. 
	\]
	
	For two disjoint subsets $S, T \subset V$, 
	let us define $f(S,T) := \sum_{u \in S} \sqrt{P(u, T)}$.
	Note that $\varphi_\frac12(A) = \frac{f(A,B)}{|A|}$, 
	so our goal is to lower bound 
	\[
	f(A,B) = \sum_{i = 1}^k f(A_i,B).
	\]
	
	For two sets $S,T \in \{A_i\}_{i = 1}^k \cup \{B_i\}_{i = 1}^k$, let us define the contiguous block $[S, T]$ to be the block of sets from $S$ clockwise up until $T$, possibly going around. 
	For example, $[B_k, A_2] := B_k \cup A_1 \cup B_1 \cup A_2$, and note that $[S,T] \neq [T,S]$ since the sets are counted clockwise. 
	
	After we set up the notations,
	we start with a lower bound on $f(A_i, B)$ by a natural function, 
	the logarithm of the size of contiguous sets,
	which is the ``slightly different lower bound'' that we mentioned before this subsection.
	
	\begin{lemma} \label{lem:log}
	For $1 \leq i \leq k$,
	\begin{eqnarray*}
	\sqrt{2C} \cdot f(A_i, B) \geq & \sum_{j = 1}^{k} \Big( \log \big(|[A_i, A_j]| + 1\big)  + \log \big(|[B_i, B_{j-1}]| + 1\big)  \\ 
	& \quad \quad - \log \big(|[B_i, A_j]| + 1\big) - \log \big(|[A_i, B_{j - 1}]| + 1\big) \Big),
	\end{eqnarray*}
	where $C$ is the normalizing constant in \autoref{def:counterexample}
	and $|[S,T]|$ denotes the number of vertices in the block $[S,T]$.
	\end{lemma}
	
	\begin{proof}
	We prove the statement for $f(A_1,B)$. 
	By definition of $f, A_i, B_i$ stated above,
	\begin{eqnarray*}
	\sqrt{2C} \cdot  f(A_1, B) 
	& = & \sqrt{2C} \cdot \sum_{i \in A_1} \sqrt{ \sum_{l = 1}^k P(i, B_l)} \\
	& = & \sqrt{2C} \cdot  \sum_{i = 1}^{a_{1}} \sqrt{\sum_{l = 1}^k \sum_{j = 1}^{b_l} P\bigg(i, \sum_{m = 1}^{l} a_m + \sum_{m = 1}^{l-1} b_m + j \bigg)}.
	\end{eqnarray*}
	By the definition of $P$ in \autoref{def:counterexample}, $P(i,j) \geq \frac{1}{C|i-j|^3}$ and so
	\begin{eqnarray*}
		\sqrt{2C} \cdot  f(A_1, B) 
		& \geq &
		\sqrt{2} \cdot \sum_{i = 1}^{a_{1}} \sqrt{\sum_{l = 1}^k \sum_{j = 1}^{b_l} \bigg(\sum_{m = 1}^{l} a_m + \sum_{m = 1}^{l-1} b_m + j - i\bigg)^{-3}} 
		\\
		& = &
		\sqrt{2} \cdot \sum_{i = 0}^{a_{1} - 1} \sqrt{\sum_{l = 1}^k \sum_{j = 1}^{b_l} \bigg(\sum_{m = 2}^{l} a_m + \sum_{m = 1}^{l-1} b_m + j + i\bigg)^{-3}}.
	\end{eqnarray*}
	We lower bound the inner sum by an integral, so that
	\begin{eqnarray*}
		\sqrt{2C} \cdot  f(A_1, B) 
		& \geq &
		\sqrt{2} \cdot \sum_{i = 0}^{a_{1} - 1} \sqrt{\sum_{l = 1}^k \int_1^{b_l + 1} \bigg(\sum_{m = 2}^{l} a_m + \sum_{m = 1}^{l-1} b_m + x + i\bigg)^{-3} \, dx}
		\\
		& = & \sum_{i = 0}^{a_{1} - 1} \bigg(\sum_{l = 1}^k \underbrace{\bigg(\sum_{m = 2}^{l} a_m + \sum_{m = 1}^{l-1} b_m + i + 1\bigg)^{-2}}_{\alpha_l} \\
		& & \quad \quad - \underbrace{\bigg(\sum_{m = 2}^{l} a_m + \sum_{m = 1}^{l} b_m + i + 1\bigg)^{-2}}_{\beta_l} \bigg)^{1/2}. 
	\end{eqnarray*}
	Now we use the following simple inequality about decreasing numbers.
	\begin{claim}
		Let $(\alpha_i)_{i=1}^k, (\beta_i)_{i=1}^k$ be positive real numbers such that $\alpha_1 \geq \beta_1 \geq \alpha_2 \geq \beta_2 \geq \dots \geq \alpha_k \geq \beta_k \geq 0$. Then
		\[
		\sum_{i = 1}^k \big(\alpha_i^2 - \beta_i^2\big) \geq \Big(\sum_{i = 1}^k \big( \alpha_i - \beta_i\big) \Big)^2.
		\]
	\end{claim} 
	\begin{proof}
		The proof is by induction. 
		For $i = 1$, the claim is clear as $\alpha_1^2 - \beta_1^2 \geq (\alpha_1 - \beta_1)^2$. 
		Suppose that the claim is true for $i = k$. 
		Let $A = \sum_{i = 2}^{k+1} (\alpha_i^2 - \beta_i^2)$ and $B = \sum_{i = 2}^{k+1} (\alpha_i - \beta_i )$.
		For the induction step, 
		we need to show that  $\alpha_1^2 - \beta_1^2 + A \geq (\alpha_1 - \beta_1 + B)^2$. 
		Since $A \geq B^2$ by induction, it suffices to show that $\alpha_1^2 - \beta_1^2 \geq  (\alpha_1 - \beta_1) (\alpha_1 - \beta_1 + 2B)$, which is equivalent to $\beta_1 \geq B$. 
		It follows from the property of decreasing sequence that $B \leq \alpha_2 \leq \beta_1$, verifying the induction step.
	\end{proof}
	
	The $\sqrt{\alpha_l}$ and $\sqrt{\beta_l}$ in the right hand side above satisfy the assumptions of the claim, and thus 
	\[
	\sqrt{2C} \cdot f(A_1, B) 
	\geq \sum_{l = 1}^k \sum_{i = 0}^{a_{1} - 1} \bigg( \bigg(\sum_{m = 2}^{l} a_m + \sum_{m = 1}^{l-1} b_m + i + 1\bigg)^{-1} - \bigg(\sum_{m = 2}^{l} a_m + \sum_{m = 1}^{l} b_m + i + 1\bigg)^{-1} \bigg)
	\]
	We again lower bound the inner sum by an integral so that $\sqrt{2C} \cdot f(A_1,B)$ is at least
	\begin{eqnarray*}
		& & \sum_{l = 1}^k \int_0^{a_1} 
		\bigg( \bigg(\sum_{m = 2}^{l} a_m + \sum_{m = 1}^{l-1} b_m + x + 1\bigg)^{-1} - \bigg(\sum_{m = 2}^{l} a_m + \sum_{m = 1}^{l} b_m + x + 1\bigg)^{-1} \bigg) \, dx
		\\
		& = & \sum_{l = 1}^k \Bigg( \log \bigg(\sum_{m = 1}^{l} a_m + \sum_{m = 1}^{l-1} b_m + 1\bigg) - \log \bigg(\sum_{m = 2}^{l} a_m + \sum_{m = 1}^{l-1} b_m + 1\bigg) 
		\\
		& & \quad \quad \quad - \log \bigg(\sum_{m = 1}^{l} a_m + \sum_{m = 1}^{l} b_m + 1\bigg) + \log \bigg(\sum_{m = 2}^{l} a_m + \sum_{m = 1}^{l} b_m + 1\bigg) \Bigg)
		\\
		& = & \sum_{l = 1}^k \Big( \log \big(|[A_1, A_l]| + 1\big) - \log \big(|[B_1, A_l]| + 1\big) - \log \big(|[A_1, B_l]| + 1\big) \\
		& & \quad \quad + \log \big(|[B_1, B_l]| + 1\big) \Big),
	\end{eqnarray*}
	using the definition e.g. $|[A_1, B_l]| = \sum_{j = 1}^{l} (a_j + b_j)$.  
	\end{proof}
	
	Next, we are going to sum up the lower bounds in \autoref{lem:log} to obtain a lower bound on $f(A,B)$.
	To write the sum nicely,
	we use a simple observation on the signs of the logarithm in our sum. 
	Let us call a contiguous block $[S, T]$ odd if there are an odd number of sets in $\{A_i\}_{i = 1}^k \cup \{B_i\}_{i = 1}^k$, and even otherwise. 
	Note that the odd blocks are exactly those with the first and last sets from the same partition $A$ or $B$, e.g. $[A_i, A_j], [B_i,B_j]$. 
	With this definition, the lower bound on $f(A,B)$ can be written as follows.
	
	\begin{lemma} \label{lem:log-sum}
	Using the definitions and notations in this subsection,
	\[\sqrt{2C} \cdot f(A, B) \geq \sum_{{S \neq T}: [S,T] \textrm{odd}} \log \big(|[S,T]| + 1\big) - \sum_{S \neq T: [S,T] \textrm{even} } \log \big(|[S,T]| + 1\big) - (k - 1) \log (n + 1),
	\]
	where the sum is over $S,T \in \{A_i\}_{i = 1}^k \cup \{B_i\}_{i = 1}^k$.
	\end{lemma}
	\begin{proof}
	We sum the inequalities in \autoref{lem:log} from $1 \leq i \leq k$.
	On the right hand side of the inequality in \autoref{lem:log}, 
	we see that all contiguous blocks starting from $A_i$ or $B_i$ are in the sum,
	with the odd blocks positive and even blocks negative. 
	Thus, summing over all $A_i$, every contiguous block is counted once as it is uniquely determined by the starting and ending sets, except for the whole cycle which appears once on the right hand side for every $i$ with a negative sign.
	\end{proof}

	To prove a lower bound on the right hand side of \autoref{lem:log-sum},
	the idea is to use the following concavity property.
	
	\begin{lemma} \label{lem:concavity}
	For $k \geq 2$, consider the function
	\begin{eqnarray*}
	h: & (a_1,b_1 \dots, a_k, b_k) \rightarrow &  \\
	& & \mkern-18mu  \mkern-18mu \mkern-18mu \mkern-18mu \mkern-18mu \mkern-18mu \mkern-18mu \mkern-18mu \sum_{S \neq T: [S,T] \text{odd}} \log \big( |[S,T]| + 1 \big) - \sum_{S \neq T: [S,T] \text{even}} \log \big(|[S,T]| + 1\big) -( k - 1) \log (n + 1),
	\end{eqnarray*}
	where the sum is over $S,T \in \{A_i\}_{i = 1}^k \cup \{B_i\}_{i = 1}^k$ and so $|[S,T]|$ depends on $a_1, b_1, \ldots, a_k, b_k$.
	Then, for all positive $j$, 
	the function 
	\[
	g: x \rightarrow h(x, b_1, s - x, b_2, \dots, a_k, b_k),
	\] 
	obtained by fixing non-negative integers $b_1, b_2, a_3, b_3, \dots, a_k, b_k$ as the size of the other sets and $s$ as the sum of $a_1 + a_2$, 
	is concave on $x \in [0, s]$.
	\end{lemma}
	
	\begin{proof}
	To prove concavity, we use the second derivative test,
	where $g$ is concave if the second derivative $g''$ is non-positive. 
	We write $g$ as $g_0 + g_1(x) + g_2(x)$, 
	where the $g_1(x)$ consists of all the log terms which contain $A_1$ but not $A_2$, and similarly $g_2(x)$ consists of all the log terms which contain only $A_2$ but not $A_1$. 
	The remaining terms are in $g_0$, which either contain both $A_1$ and $A_2$ or none of $A_1$ and $A_2$.
	Note that these terms are independent of $x$, because if a block $[S,T]$ contains both $A_1$ and $A_2$ then its size $|[S,T]|$ is the same even when we change $x$, so these terms can be ignored when we compute derivatives.
	
	Let us focus on $g_1(x)$ first.
	The blocks that contain $A_1$ but not $A_2$ must be of the form $[S,A_1]$ or $[S,B_1]$ for some set $S$.
	Let $\sigma([S,T])$ denote the parity of the block $[S,T]$.
	Note that the parity of $[S,A_1]$ and $[S,B_1]$ are different, and so
	\begin{eqnarray*}
		g_1''(x) 
		& = & 
		\sum_{S} (-1)^{\sigma([S, A_1]) + 1} \Big( \big(\log \big(|[S,A_1]| + 1\big)\big)'' - \big(\log \big(|[S,B_1]| + 1\big)\big)'' \Big)
		\\
		& = & \sum_{S} (-1)^{\sigma([S, A_1]) + 1} \Big( \log (|[S,B_k]| + x + 1)'' - (\log (|[S,B_k]| + x +  b_1 + 1) \Big)'' 
		\\
		& = &
		\sum_{S} (-1)^{\sigma([S, A_1])} \Big( \big(|[S,B_k]| + x + 1\big)^{-2} - \big(|[S,B_k]| + x + b_1 + 1\big)^{-2} \Big) 
		\\
		& = & \sum_{S} (-1)^{\sigma([S, A_1])} \bigg( b_1 \cdot \big(|[S,B_k]| + x + 1\big)^{-2} \cdot \big(|[S,B_k]| + x + b_1 + 1\big)^{-1} 
		\\ 
		&  &  
		\quad \quad \quad \quad \quad \quad \quad \quad
		+~ b_1 \cdot \big(|[S,B_k]| + x + 1\big)^{-1} \cdot \big(|[S,B_k]| + x + b_1 + 1\big)^{-2} \bigg),
	\end{eqnarray*}
	where the sum is over $S \in \{A_i\}_{i = 1}^k \cup \{B_i\}_{i = 1}^k$.
	In the special case when $S = A_1$, we violate our own notation and let $|[A_1,B_k]| = 0$ in this proof; all other cases are still the same.
	
	When $b_1 = 0$, the sum equals zero and we are done, so assume $b_1 \neq 0$. To see that $g_1''(x)$ is negative, we pair up the terms with $S=B_i$ and $S=A_{i+1}$ with indices taken modulo $k$ so that
	\begin{eqnarray*}
		\frac{1}{b_1} \cdot g_1''(x) 
		& = &
		\sum_{i = 1}^k \Big[ \big(|[B_i,B_k]| + x + 1\big)^{-2} \cdot \big(|[B_i,B_k]| + x + b_1 + 1)^{-1}  
		\\
		& & \quad \quad +~\big(|[B_i,B_k]| + x + 1\big)^{-1} \cdot \big(|[B_i,B_k]| + x + b_1 + 1\big)^{-2}
		\\
		& & \quad \quad -~\big(|[A_{i + 1},B_k]| + x + 1\big)^{-2} \cdot \big(|[A_{i + 1},B_k]| + x + b_1 + 1\big)^{-1}
		\\
		& & \quad \quad -~\big(|[A_{i + 1},B_k]| + x + 1\big)^{-1} \cdot \big(|[A_{i + 1},B_k]| + x + b_1 + 1\big)^{-2} \Big]
		\\
		& < & 0,
	\end{eqnarray*}
	where the inequality holds because $|[A_{i+1},B_k]| < |[B_i,B_k]|$ and so each summand is negative (recall the special case that $|[A_1,B_k]|=0$ in this proof).
	
	The function $g_2(x)$ is handled analogously in view of the symmetry of the second derivative of the logarithm. 
	This proves that $g$ is concave. 
	\end{proof}
	
	With the concavity property,
	we can apply a simple ``swapping/merging'' argument to reduce to the case when there is only one contiguous set, i.e. $k=1$, and then finish the proof.
	
	By concavity, the function $g(x)$ attains its minimum at one of the endpoints, and so
	\[
	h(a_1, \dots, b_n) \geq \min \big\{ h(0, b_1, a_1 + a_2, b_2, \dots, a_n, b_n),h(a_1 + a_2, b_1, 0, b_2, \dots, a_n, b_n) \big\}
	\]
	The next observation is that when one set has size zero, we can merge the two adjacent sets in the same partition into one.
	More formally, let $b_1=0$ without loss of generality, we claim that 
	\[
	h(a_1, 0, a_2, b_2 \dots, a_k, b_k) = h(a_1 + a_2, b_2, \dots, a_k, b_k).
	\]
	To see this, note that $|[S,A_1]|$ and $|[S,B_1]|$ have the same values but they have different signs so the terms involving them cancel out each other, and similarly the terms involving $|[B_1,S]|$ and $|[A_2,S]|$ cancel out each other.
	Therefore, in $h(a_1, 0, a_2, b_2, \ldots, a_k, b_k)$, there are no terms ending with $A_1$ or $B_1$ and no terms beginning with $B_1$ or $A_2$, and all the remaining terms have a one-to-one correspondence with the terms in $h(a_1 + a_2, b_2, \ldots, a_k, b_k)$.
	
	This reduces $k$ by one.  
	Repeating the same argument until $k=1$, we see that
	\[
	\sqrt{2C} \cdot f(A,B) 
	\geq h(|A|, n-|A|)
	= \log (|A| + 1) + \log (n - |A| + 1) - \log (n + 1),
	\]
	and thus 
	\begin{eqnarray*}
	\varphi_\frac12(G) 
	= \min_{A: |A| \leq \frac{n}{2}} \frac{f(A,B)}{|A|} 
	&\gtrsim & \min_{l: l \leq \frac{n}{2}} \frac{h(l,n-l)}{l} \\
	& = & \min_{l: l \leq \frac{n}{2}} \frac{\log (l + 1) + \log (n - l + 1) - \log (n + 1)}{l},
	\end{eqnarray*}
	where we used that $C$ is upper bounded by an absolute constant.
	
	It remains to lower bound the right hand side.
	Since $l \leq \frac{n}{2}$, it follows that 
	$\log(n - l - 1) - \log(n+1) \geq  \log((n+1)/2) - \log(n+1) = -\log2$,
	and so
	\[
	\frac{\log (k + 1) + \log (n - k + 1) - \log (n + 1)}{k} 
	\geq \frac{\log \frac{k+1}{2} }{k} 
	\gtrsim \frac{\log \frac{k+1}{2} }{\frac{k+1}{2}} 
	\geq \frac{\log n}{n},
	\]
	where the last inequality is because $\frac{\log n}{n}$ is decreasing for $n \geq 3$ and for $1 \leq k \leq 4$ the last inequality clearly holds when $n$ is large enough.
	This concludes the proof of \autoref{thm:counterexample}.

	\subsection*{Concluding Remarks and Open Questions}
	
	We believe that the same analysis of $\varphi_p(G)$ can be extended to other generalizations of Cheeger's inequality in~\cite{LOT12,KLLOT13}, and also to the directed edge conductance using the recent notions of reweighted eigenvalues in~\cite{OZ22, KLT22, LTW23, LTW24}. We leave it as an open question to find a counterexample where the transition matrix is the simple random walk matrix of a graph.

        \subsection*{Acknowledgements}

        We thank Christian Houdr\'e and Prasad Tetali for their encouragement and the anonymous reviewers for their helpful comments.
 
	\bibliography{talagrand_ref}
\end{document}